\documentclass{article}  

\newenvironment{tightList}
{\begin{list}{}{\partopsep=\baselineskip
	\parskip=0pt
	\parsep=0pt
	\topsep=0pt
	\itemsep=0pt
	\labelwidth=0pt
	\itemindent=-10pt}}
{\end{list}}

\usepackage{amsmath,amsthm,amsfonts,amssymb,verbatim}
\newcommand{\Real}{\mathbb{R}}

\newtheorem{theorem}{Theorem}[section]
\newtheorem{definition}[theorem]{Definition}
\newtheorem{Lemma}[theorem]{Lemma}
\newtheorem{lemma}[theorem]{Lemma}

\newcommand{\sA}{\mathcal{A}}
\newcommand{\sC}{\mathcal{C}}
\newcommand{\sD}{\mathcal{D}}
\newcommand{\sF}{\mathcal{F}}
\newcommand{\sI}{\mathcal{I}}
\newcommand{\sS}{\mathcal{S}}

\usepackage{url}

\usepackage{datetime}
\yyyymmdddate

\usepackage{amsmath}

 \usepackage{graphicx}

\newcommand{\pmf}{p.m.\ }
\newcommand{\pmfNoSpace}{p.m.}

\newcommand{\mlf}{m.l.f.\ }
\newcommand{\mlfNoSpace}{m.l.f.}
\newcommand{\mlfNoPeriod}{m.l.f}

\newcommand{\indexed}[1]{\{{#1}\ldots\}}
\newcommand{\indexedFull}[3]{\{{#1}_{{#2}}\ |\ {#2}\in\mathcal{{#3}}\}}

\DeclareMathOperator{\ttv}{\mathit{TTV}}
\DeclareMathOperator{\tv}{\mathit{TV}}

\usepackage{mathrsfs}  
\newcommand{\vt}[2]{\mathscr{V}({#1},{\mathcal{#2}})}

\newcommand{\vf}[3]{\mathscr{F}({#1},{\mathcal{#2}},{\mathcal{#3}})}

\makeatletter
\renewcommand*{\@fnsymbol}[1]{\ensuremath{\ifcase#1\or \ \or *\or \dagger\or \ddagger\or
   \mathsection\or \mathparagraph\or \|\or **\or \dagger\dagger
   \or \ddagger\ddagger \else\@ctrerr\fi}}
\makeatother

\title{Varilets: \\Additive Decomposition,\\Topological Total Variation, and Filtering \\ of Scalar Fields}
\author{Martin Brooks \\ Apollo Systems Research Corporation \\ \emph{martin.fraser.brooks@icloud.com} \thanks{\copyright 2015 Martin Brooks}}

\begin{document}       
\maketitle 
\begin{abstract}
Continuous interpolation of real-valued data is characterized by \emph{piecewise monotone} functions on a compact metric space.
\emph{Topological total variation} of piecewise monotone function $f:X \to \Real$ is a homeomorphism-invariant  generalization of 1D total variation.
A \emph{varilet basis} is an orthonormal collection of piecewise monotone functions $\{g_i\ |\  i = 1 \ldots n\}$, called \emph{varilets}, such that every linear combination $\sum a_ig_i$ ($a_i \in \Real$) has topological total variation $\sum |a_i|$. A \emph{varilet transform} for $f$ is a varilet basis for which $f =\sum \alpha_ig_i$. Filtered versions of $f$ result from altering the coefficients $\alpha_i$.
\end{abstract}

\newpage
\tableofcontents

\section{Introduction}

Topology has proved to be a powerful tool for discovery of the structure of data \cite{EdelsbrunnerBook, Carlsson2009}. Typically, data is viewed as a finite sample of a continuous function. The  function becomes the object of study, often via  a topological representation such as the contour tree \cite{Carr2000}, Reeb graph \cite{Reeb1946}, Morse-Smale complex  \cite{Bremer2004}, persistence diagram \cite{EdelsbrunnerBook} or persistence barcode \cite{Carlsson2004}. 

This paper presents new results for topological analysis of continuous functions.
Whereas computational topology often utilizes algebraic topology \cite{EdelsbrunnerBook}, we exploit the older field of analytic topology \cite{Whyburn1942}.

Our starting point is real-valued continuous functions on compact metric spaces.
We introduce the broad category of \emph{piecewise monotone} functions,  well-suited to data interpolation.
We introduce  a new measure on functions -- \emph{topological total variation} -- and we introduce \emph{varilets}, an orthonormal basis which additively decomposes both the function and its topological total variation.

The \emph{varilet transform} maps a function to a varilet basis by means of a \emph{lens} parameter.
A lens is collection of upper and lower level set components, providing a multiresolution view of the function.  

We state a mathematical algorithm -- the Varilet Transform Algorithm -- and prove its output correct.


We proceed with an overview, followed by discussion of related work and identification of contributions.  
Two sections then develop the analytic topology, the first defining the Varilet Transform Algorithm, and the second validating its correctness.

\subsection{Overview}
\label{overview}

\emph{Analytic topology} \cite{Whyburn1942},  \emph{continuum theory} \cite{Whyburn1942, HY1961, Nadler1992, Charatonik1998} and \emph{dynamic topology} \cite{Whyburn1970, DW1979} had their heyday in the mid-twentieth century. The  Eilenberg-Whyburn \emph{monotone-light factorization} \cite{Lord1997, walker1974} is a powerful result concerning functions on compact metric spaces. 

This paper uses the monotone-light factorization as the foundation for topological analysis of real-valued functions. 
Varilets are an elementary application of analytic topology. 

Continuous function $f:X \to Y$ is \emph{monotone} when $f^{-1}y$ is connected for all $y \in Y$; ; thus $f^{-1}$ carries connected sets to  connected sets.
$f$ is \emph{light} when $f^{-1}y$ is totally disconnected for all $y \in Y$.
The monotone-light factorization \cite{Whyburn1942, Lord1997} states that there exists a unique compact metric space $M$ -- called $f$'s \emph{middle space} -- such that  $f =  \lambda \circ  \mu$, where $\mu :X \to M$ is monotone and $\lambda :M  \to Y$ is light.

The middle space $M$ of function $f:X \to Y$ is the quotient of the domain that identifies all points which, for some $y \in Y$, lie in the same connected component of  $f^{-1}y$. $f$'s monotone factor $\mu$ is the quotient map from $X$ to $M$; $f$'s light factor $\lambda$ assigns each point $p \in M$ the value $f(\mu^{-1}p) \in Y$. Thus $f = \lambda \circ  \mu$. 

Piecewise monotonicity is defined in terms of the middle space. The middle space $M$ of \emph{piecewise monotone} scalar field $f:X \to \Real$ is a graph having finitely many vertices and edges, with each edge having the topology of a closed real interval. $M$ is identical to $f$'s Reeb graph \cite{Reeb1946}.  The monotone factor $\mu$ locates extrema, saddles  and contours in the domain $X$; the light factor $\lambda$ provides their numerical values. 

The monotone-light factorization enables one to proceed by first defining constructions using middle space $M$ and light factor $\lambda$, and then applying $\mu^{-1}$ to pull back to the domain $X$ and function $f$. By restricting the middle space to a finite graph, these constructions enjoy the simple topology of graph continua \cite{Nadler1992}. Monotonicity of $\mu$ makes the theory oblivious to the many complexities of continua.

We use  the light factor $\lambda$ to measure length along the edges of $M$: For points $a, b$ on an edge, the length between them is $|\lambda(a) - \lambda(b)|$.
\emph{Topological total variation} $\ttv(f)$  is the sum of all $M$'s edge lengths. 

A \emph{varilet basis} is a finite collection $\{g_i\ |\ i = 1 \ldots n\}$ of real-valued piecewise monotone functions, called \emph{varilets}, such that every linear combination $\sum a_ig_i$ (each $a_i \in \Real$) has topological total variation $\ttv(\sum a_ig_i) = \sum |a_i|$. A varilet basis is normal  and independent in the sense that each $\ttv(g_i) = 1$ and $\sum a_ig_i = \sum b_ig_i$ only when each $a_i = b_i$.

The name ``varilet'' reflects the relationship of the basis functions to topological total variation. 

A \emph{varilet transform} for piecewise monotone $f$ is a varilet basis such that $f$ is a linear combination:  $f = \sum \alpha_ig_i$. 
The Varilet Transform Algorithm is stated as a mathematical algorithm from which computational methods may be derived.

\emph{Varilet filters} for $f$ are created by varying the coefficients, yielding filtered functions $f' = \sum a_ig_i$. Varilet filters manipulate topological total variation analogously to linear filters' manipulation of energy.

\subsection{Related Work}

Varilets fit into the larger context of computational topology and data analysis \cite{EdelsbrunnerBook, Carlsson2009}.
Although there may exist mathematical connections to persistent homology \cite{Edelsbrunner2002} and discrete Morse theory \cite{Forman01}, this paper focuses instead on the monotone-light factorization \cite{Whyburn1942}, which for piecewise monotone functions may be seen as a decorated version of the Reeb graph \cite{Reeb1946}. Sometimes called the contour tree \cite{Carr2000}, and within persistence theory the merge tree \cite{EdelsbrunnerBook}, the Reeb graph has been often used for simplification of scalar fields \cite{Carr2004, Tierny2012} and has been exploited throughout computational topology. 

The present paper has connections to the work of Bauer et al. on persistence and total variation \cite{Bauer2010}, and to work relating topological analysis to signal processing, including Guillemard et al. \cite{Guillemard2013, Guillemard} and Bauer et al. \cite{Bauer2014}. There are similarities in intent, but not in formalism, to Robinson's notion of topological filter \cite{Robinson}.

Our definition of topological total variation is similar to graph total variation as defined by Bresson \cite{graphTV}. Our definition agrees with the usual 1D definition of total variation, but conflicts with most multidimensional definitions \cite{TVdefs}, including that used in image processing \cite{Rudin}.

\subsection{Contributions}

This paper introduces the varilet transform, an additive decomposition of scalar fields by independent normalized summands, which also additively decomposes a generalized total variation measure.

The Varilet Transform Algorithm provides a mathematical skeleton for computational methods. 
The algorithm is proved correct.

Filtering $f$ by varying the coefficients of its varilet transform provides an generalization of simplification of scalar fields \cite{Carr2004, Tierny2012}.

\subsection{Organization}

This paper has two main sections: Section \ref{defs} provides definitions, culminating in the Varilet Transform Algorithm. Section \ref{validation} validates the algorithm, proving that it computes a varilet basis.

\section{Definition of the Varilet Transform Algorithm}
\label{defs}

Following some preliminaries, we define piecewise monotonicity, topological total variation, and then varilet bases, transforms and filters.
Finally, we state the Varilet Transform Algorithm. 

\subsection{Preliminaries}

All spaces in this paper are compact metric spaces.
A continuum is a connected component of a compact metric space; we assume all continua are non-degenerate (no isolated points).
Continua include line segments, disks, spheres, simplexes, graphs having one-dimensional topology on their arcs (e.g. Reeb graphs), compact manifolds, as well the result of (appropriately) attaching together other continua. 

We specify $f$'s monotone-light factorization, abbreviated ``\mlfNoSpace'', by simply listing $\mu M\lambda$, where $\mu$ is the monotone factor, $M$ is the middle space, and $\lambda$ is the light factor.

For continuous $f$ on compact metric space $X$, the middle space $M$ has finitely many connected components; they are  in 1-1 correspondence with $X$'s components. There are no relations among components of $X$,  nor among components of $M$. 

$f$'s middle space $M$ is a compact metric space;  therefore  assertions proven for $f$ also apply to $\lambda$.
The \mlf of $\lambda$ is $1M\lambda$, where $1$ represents the identity function on $M$.

We refer to the connected components of $f^{-1}y$ as \emph{contours}.

For any subset $S$ of a topological space, we denote the interior by $S^{\circ}$, the closure by $\overline{S}$, the boundary by $\partial S$, and the complement by $S^c$.

\subsection{Piecewise Monotone Functions}

This section introduces \emph{piecewise monotone} functions, which will be used throughout. They correspond in principle to use of \emph{tame} functions \cite{Edelsbrunner2007,EdelsbrunnerBook}. 

The monotone-light factorization provides the basis for piecewise monotonicity.  We give a general definition, followed by specialization to real-valued functions.

\begin{definition}[Piecewise Monotone Function]
\label{pmf}
Suppose $f:X \to Y$ has \mlf $\mu M\lambda$.
We say that $\lambda$ is \emph{locally monotone}  at $p \in M $ when $p$ has a neighbourhood upon which $\lambda $ is monotone.
Let $M ^*$ denote the set of all points of $M$ at which $\lambda$ is locally monotone.
Then $f$ is \emph{piecewise monotone} when:
\begin{tightList}
\item[(1)] $M ^*$ is dense in $M$;
\item[(2)] $M ^*$ has finitely many components; and
\item[(3)] $\lambda $ is monotone on the closure of each component of $M ^*$.
\end{tightList}
\end{definition}

The closures of the components of $M ^*$ are the \emph{monotone pieces} referred to in the name, which we abbreviate as ``\pmfNoSpace''.

The light factor $\lambda$ is a homeomorphism on each monotone piece. 

When $f:[0\ 1] \to \Real$, definition \ref{pmf} provides the usual meaning of  piecewise monotone. 

For \pmf scalar field $f:X \to \Real$, each monotone piece of the middle space $M$ is a closed, non-degenerate real interval. Some interval end points are shared between two or more -- but only finitely many -- intervals; these are the points of definition \ref{pmf} at which $\lambda$ fails to be locally monotone. Some interval endpoints are not shared -- in this case $\lambda$ is locally monotone at the endpoint. The monotone pieces form a finite graph; the interval endpoints are the graph vertices; each interval defines a single edge. $M$ may have more than one edge between a pair of vertices, and multiple edges may constitute a loop, but $M$ does not have an edge connecting a vertex to itself. 

The middle space $M$ of \pmf $f:X \to \Real$ is graph-theoretically and topologically identical to $f$'s Reeb graph \cite{Reeb1946}. We will use both its graph structure and its point-set topology, but we will not be concerned with homotopy. When compact metric space $X$ has connected components $X_1 \ldots X_n$, then the middle space $M$ comprises $n$ disjoint graphs $M_1 \ldots M_n$, and each $f|X_i$ is piecewise monotone with middle space $M_i$.


It is well known that when $X$ is simply connected then $M$ is acyclic; when $f:[0\ 1] \to \Real$ then $M$ is linked chain of closed  intervals. 
Graph continua appear in the literature \cite{Nadler1992, Georgakopoulos}, but have not been previously utilized in relation to monotone-light factorization.

For $f:X \to \Real$, the light factor $\lambda$ is numerically strictly monotone along each edge of $M$. For an edge having endpoint at vertex $V$, we use $\lambda$ to characterize the edge as \emph{increasing} or \emph{decreasing} at $V$. If vertex $V$ has both an increasing and a decreasing edge, then it is a \emph{saddle}; in this case $V$ terminates three or more edges. A vertex for which all edges have the same direction is an \emph{extremum}, either a maximum or minimum. The extrema and saddles of $M$ are called \emph{critical points}. (This usage is more general than the classical notion, in the same spirit as homological critical values in \cite{Edelsbrunner2007}.)

We use  $\lambda$ to measure length along the edges of $M$: For any two points $a, b$ on an edge (including endpoints), the length between them is $|\lambda(a) - \lambda(b)|$. 

When $f$ is piecewise monotone, then so is its light factor $\lambda$.

\subsection{Topological Total Variation}

Total variation  $\tv(f)$ for differentiable $f:[0\ 1] \to \Real$ is given by $\tv(f) = \int |f'|$.  On multidimensional domain $X \subset \Real^n$, total variation has a gradient formulation $\tv(f) = \int |\nabla{f}|$. For non-differentiable functions, see definitions, examples and historical discussion in \cite{TVdefs}. 


We provide an alternative definition:

\begin{definition}\label{ttv}
For \pmf  $f:X \to \Real$ having \mlf $\mu M\lambda$, $f$'s \emph{topological total variation} $\ttv(f)$  is the sum of all $M$'s edge lengths.
\end{definition}

$\ttv(f) = \tv(f)$ when $f: [0\ 1] \to \Real$. 

However, when $X$ is multidimensional then $\ttv$ and $\tv$ do not agree, seen as follows. The co-area formula  \cite{TVdefs} in equation  (\ref{per}) expresses total variation  as the integral of level set perimeter lengths $\mathit{Per}(f, y)$,

\begin{equation}
\label{per}
\tv(f) = \int_{- \infty}^{\infty} \mathit{Per}(f, y) \, dy.
\end{equation}

Whereas $\ttv$ is invariant under self-homeomorphisms of $X$,  
equation (\ref{per}) indicates that multidimensional $\tv$ is not; for example, consider a homeomorphism of $X$ that stretches $f$'s level set perimeters.

Topological total variation avoids dependence on the particulars of $f$'s domain $X$ by measuring instead on $f$'s middle space $M$. 
$\ttv(\lambda) = \ttv(f)$.

When compact metric space $X$ has connected components $X_1 \ldots X_n$, then $\ttv(f) = \sum \ttv(f|X_i)$.

\subsection{Varilet Basis}
 
\begin{definition}\label{varilet}
A finite collection $\{g_i:X \to \Real\ |\ i = 1 \ldots n\}$ of \pmf functions is a \emph{varilet basis} when $\ttv(\sum a_ig_i) = \sum |a_i|$ for all choices of $a_i \in \Real$.
 \end{definition}
 
 The functions $g_i$ are called \emph{varilets}.
 A varilet basis is normal and independent in the sense that each $\ttv(g_i) = 1$, and $\ttv(\sum a_ig_i) \equiv 0$ iff each all $a_i = 0$, implying that every linear combination is unique. 
 
\subsection{Varilet Transforms}

 \begin{definition}
 \label{varilet_transform_def}
A \emph{varilet transform} for \pmf $f:X \to \Real$ is a varilet basis $\{g_i\ |\ i = 1 \ldots n\}$ such that there exist positive \emph{amplitudes} $\{\alpha_i\ |\ i = 1 \ldots n\}$ giving $f = \sum \alpha_ig_i$.
 \end{definition}
 
Unlike the Fourier transform, there is no unique varilet transform for a piecewise monotone function; instead there are multiple varilet transforms, in this respect similar to wavelet transforms. 

Section \ref{varilet_transform_section} provides a mathematical algorithm having two input parameters: the function $f$, and a choice of \emph{lens} $\sC$. The Varilet Transform Algorithm outputs a varilet basis $\{g_i\ |\ i = 1 \ldots n\}$ and amplitudes $\{\alpha_i\ |\ i = 1 \ldots n\}$. Section \ref{validation} will validate the algorithm by developing some elementary analytic topology. 

The Varilet Transform Algorithm provides one way to compute varilet transforms, but definition \ref{varilet_transform_def} may admit other methods as well.
  
\subsection{Varilet Filters}

\begin{definition}
Let $\{g_i\ |\ i = 1 \ldots n\}$ be the varilet basis resulting from a varilet transform for $f$.
Every choice of \emph{varilet filter coefficients} $\{a_i \in \Real |\ i = 1 \ldots n\}$ defines a \emph{filtered version of $f$}, $f' = \sum a_ig_i$. 
\end{definition}

Since $\{g_i\ |\ i = 1 \ldots n\}$ is a varilet basis, it follows that $\ttv(f') = \sum |a_i|$.


\subsection{Varilet Transform Algorithm}
\label{varilet_transform_section}

This section introduces the mathematical objects and algorithm for computing varilet transforms.
Objects introduced are the \emph{varilet lens} $\sC$, and \emph{varilet supports} $\sD$.
The lens $\sC$ plays the role of a parameter in the algorithm.

\subsubsection{Varilet lens}

The Varilet Transform Algorithm analyzes a function, the parameter $f$.
This section introduces the only additional parameter, the transform's \emph{lens} $\sC$,
so-called because it determines the selection and resolution of $f$'s analysis in relation to critical points in $f$'s middle space $M$. 

We do not specify how to choose the lens $\sC$, but instead take an axiomatic approach, relating specialized properties of $\sC$ to resulting properties of the varilet transform.
Looking forward to those results, the varilet transform's twin additive decomposition of function $f$ and topological total variation $\ttv(f)$ is a \emph{universal} property of the varilet transform, i.e.\ it holds for all $\sC$. Whereas, in the contexts of image segmentation, simplification and fractal analysis \cite{varilet_image}, we will require lens $\sC$ to have special properties.


\begin{definition}[Varilet Lens]
\label{varilet_lens}
A subset $C \subset M$ is a \emph{constant-boundary region} when $C$ is closed, connected, has nonempty interior, and $\lambda$ is constant on $\partial C$. 

Let $\sI$ be a finite sequence of natural or real numbers that will serve as indices. 
An indexed collection $\sC = \indexedFull{C}{i}{I}$  of constant-boundary regions is \emph{nested} when indices $i < j$ imply either $C_i \supsetneq C_j$ or $C_i \cap C_j= \emptyset$. 

A \emph{varilet lens for $f$} is nonempty nested collection $\sC$ of constant-boundary regions such that  $\underset{\sI}{\bigcup}C_i = M$. 

The \emph{root regions} of varilet lens $\sC$ are the connected components of $M$, and 
the \emph{root indices} are the indices of the root regions. 

For each $i \in \sI$, we define $C_i$'s \emph{successors} $\sS_i$  be the collection of all maximal (by inclusion) $C_j \subset C_i$.
When $C_j$ is a successor of $C_i$, we say that $C_i$ is the \emph{predecessor} of $C_j$. 
\end{definition}
 
When unambiguous, we suppress the index set $\sI$, writing $\indexed{C_i}$ instead.
 
The successor relation exactly reflects the tree of inclusions within $\sC$.

The root regions are the only constant-boundary regions having empty boundary.
The root indices comprise an initial subsequence of $\sI$.

The nesting structure of $\sC$ suggests a \emph{multiresolution lens}, with fine resolution for tightly nested sets in $\sC$ and coarse resolution for sparsely nested sets. 

\subsubsection{Varilet Supports}
\label{supports}

From the varilet lens $\sC$ we derive the \emph{varilet supports} $\sD$, which play a central role in the Varilet Transform Algorithm. 

\begin{definition}
\label{suppports}
Let $\sC = \indexedFull{C}{i}{I}$ be a varilet lens for $f$. Then \emph{$\sC$'s varilet supports} comprise the collection $\sD = \{D_i \ |\ i \in \sI\}$,  where each \emph{varilet support} is defined as $D_i = (\overline{C_i \smallsetminus \underset{j > i}{\bigcup}C_j})$.
\end{definition}

The name ``support'' is discussed in section \ref{algo}, after introducing the Varilet Transform Algorithm.

Due to the simplicity of the finite-graph continuum topology of $f$'s middle space $M$, we can state some elementary properties:
\begin{tightList}
\item[$\bullet$] A varilet support $D_i$ is not necessarily connected, but has only finitely many components, each of which is the closure of a nonempty connected interior. 
\item[$\bullet$] Each component of each $D_i$ is a connected graph fragment cut out of $M$; these fragments overlap only at the cut points, which may be either regular or critical points. 
\item[$\bullet$] $\sD$ covers $M$, and supports can intersect only at their boundaries.  
\item[$\bullet$] Each $C_i = \cup \{D_j\ |\ D_j \subset C_i\}$ and each $\partial C_i \subset \cup \{\partial D_j\ |\ D_j \subset C_i\}$.
\item[$\bullet$] When no $C_j \subsetneq C_i$, then $D_i = C_i$.
\end{tightList}

\subsubsection{Varilet Transform Algorithm}
\label{algo}

The  Varilet Transform Algorithm is a mathematical algorithm resulting in a varilet transform for $f$ based on lens $\sC$. 
Since there are many possible $\sC$, the algorithm can produce as many different varilet transforms for $f$. 

The output of the Varilet Transform Algorithm is an indexed collection of functions $\{g_i :X \to \Real\ |\ i \in \sI\}$, and a correspondingly indexed collection of positive reals $\indexedFull{\alpha}{i}{I}$. 

We equate the collection of functions to the \emph{varilet transform of $f$ with lens $\sC$}, written: 
\begin{equation*}
\vt{f}{C} = \indexedFull{g}{i}{I},
\end{equation*}
and we call $\indexedFull{\alpha}{i}{I}$ the transform's \emph{amplitudes}.

\begin{definition}[Varilet Transform Algorithm] 
\label{varilet_transform_algo}
Suppose \pmf $f:X \to \Real$ has \mlf $\mu M\lambda$.
Let $\sC=\indexedFull{C}{i}{I}$ be a varilet lens for $f$, and let $\sD=\indexed{D_i}$ be the varilet supports. 

Each function $g_i$ will be created by first creating functions $\lambda_i,\gamma_i:M \to \Real$,  then defining $g_i$ on $X$ via $f$'s monotone factor:  $g_i = \gamma_i \circ \mu$.

For each $i \in \sI$, in any order, or asynchronously, define $g_i$ and $\alpha_i$:
\begin{tightList}
\item
\item [(1)] Define $\alpha_i = \ttv(\lambda|D_i)$.
\newline
\newline
\emph{Note: Lemma \ref{restriction-extension} will show $\lambda|D_i$ piecewise monotone.}
\item
\item[(2)] Let $M^*$ be the component of $M$ containing $D_i$. 
\newline\newline
Define $\lambda_i$ on $M^*$ as the unique extension of $\lambda|D_i$ that is constant on each component of $M^* \smallsetminus D_i$. 
\newline
Define $\lambda_i \equiv 0$ on all other components of $M$.
\newline
\newline
\emph{Note: Lemma \ref{flat_extension_lemma} will show that there exists a unique piecewise constant extension.
Lemma \ref{restriction-extension} will show $\lambda_i$ piecewise monotone and $\ttv(\lambda_i) = \alpha_i$.}
\item
\item[(3)] Define $\gamma_i$ on $M^*$ by scaling and shifting $\lambda_i$.
\begin{equation*}
\begin{split}
\forall p &\in M^*, \\
& \gamma_i(p) =
\begin{cases}
\alpha_i^{-1} \lambda_i(p) &\text{when $i$ is a root index}\\
\alpha_i^{-1}(\lambda_i(p) - \lambda(\partial C_i)) &\text{otherwise.}
\end{cases}
\end{split}
\end{equation*}
Define $\gamma_i \equiv 0$ on all other components of $M$.
\newline
\newline
\emph{Note: Lemma \ref{restriction-extension} will show $\gamma_i$ piecewise monotone and $\ttv(\gamma_i) = 1$.}
\item
\item[(4)] Define $g_i = \gamma_i \circ \mu$.
\newline
\newline
\emph{Note: Lemma \ref{restriction-extension} will show $g_i$ piecewise monotone and $\ttv(g_i) = 1$.}

\end{tightList}
\end{definition}

The Varilet Transform Algorithm produces functions for which we have not yet proved any properties; in fact, we must prove that the algorithm creates well-defined functions.  
The next section proves  that $\vt{f}{C} = \indexedFull{g}{i}{I}$ is a varilet basis expressing $f = \sum \alpha_ig_i$.

We can now better motivate the name ``support''.
Each $D_i$ \emph{supports} functions $\lambda_i$ (and $\gamma_i$) in the somewhat nonstandard sense that $\lambda_i$  
is nowhere constant on $D_i$ but is constant on each component of $M \smallsetminus D_i$.
``Support'' usually means ``non-zero'' but here means ``non-constant''.

\newpage

\section{Validation of the Varilet Transform Algorithm}
\label{validation}
In this section we prove:

\begin{theorem}[Varilet Transform Theorem]
\label{transform} 
The result of the Varilet Transform Algorithm,
$\vt{f}{C} = \indexedFull{g}{i}{I}$, is a varilet basis, and $f = \sum \alpha_ig_i$. 
\end{theorem}

The rather involved proof compensates for a lack of preexisting mathematics by creating a small theory having its own vocabulary.  
Almost all work uses the middle space and light factor as surrogates for domain and function, using the monotone factor to connect the two when required.

Section \ref{well_defined} establishes that the Varilet Transform Algorithm is well-defined. Section \ref{additive_decomposition} shows that the algorithm additively decomposes $f$ and $\ttv(f)$. Section \ref{varilet_basis_section} shows that the algorithm creates is a varilet basis, completing proof of theorem \ref{transform}.

Throughout, we work with \pmf $f:X \to \Real$ having \mlf $\mu M\lambda$.

\subsection{Flat Extensions and Constant-Boundary Functions}
\label{functions}

We start by identifying two classes of functions on $f$'s middle space $M$. 

\begin{definition}[Flat Extension]
\label{flat-ext}
Consider any closed $D \subset M$ and any function $\pi:D \to \Real$, and let $\pi^*:M \to \Real$ be an extension of $\pi$. Then $\pi^*$ is a \emph{flat extension} when $\pi^*$ is constant on each component of $M \smallsetminus D$.
\end{definition}

Varilet lens $\sC$ causes the varilet transform and filtered versions of $f$ to be similar, in the sense that they all have form $\pi \circ \mu$, where $\pi$ is the following type of function.

\begin{definition}[Constant-Boundary Functions]
$\pi:M \to \Real$ is a \emph{constant-boundary function for $\sC$} when $\pi$ is constant on each $\partial C_i$, for $i \in \sI$.
\end{definition}

$\lambda$ is constant-boundary for $\sC$.

In this paper we will define new constant-boundary functions $\pi$, identifying useful relationships between functions $f$ and  $f' = \pi \circ \mu$, where $\pi$ is substituted for $\lambda$ in $f$'s monotone-light factorization.

$\sC$'s constant-boundary functions are closed under composition with any continuous function $\zeta:\Real^n \to \Real$, i.e.\ $\pi(p) = \zeta(\pi_1(p) \ldots \pi_n(p))$ is constant-boundary, for any choices of constant-boundary $\pi_1 \ldots \pi_n$.

In the Varilet Transform Algorithm, once we have established that the functions $\lambda_i, \gamma_i$ are well-defined, it will follow that each is  constant-boundary for $\sC$.

\subsection{Varilet Transform Algorithm Well-Defined}
\label{well_defined}

In this section we validate the notes attached to the Varilet Transform Algorithm, thereby establishing that each function $\lambda_i, \gamma_i$ and $g_i$ is well-defined, \pmfNoSpace, and has  topological total variation as indicated in the statement of the algorithm. 

Section \ref{restrictions_&_extensions} proves elementary properties of restrictions and extensions to piecewise monotone functions. 
Section \ref{flat_extension} proves existence and uniqueness of flat extensions to restrictions of constant-boundary functions, completing the proof that the Varilet Transform Algorithm is well-defined.
 
\subsubsection{Restrictions \& Extensions of Piecewise Monotone Functions}
\label{restrictions_&_extensions}

Statement (a) of the following Restriction-Extension Lemma validates step (1) the Varilet Transform Algorithm. 
Once lemma \ref{flat_extension_lemma} establishes that the functions $\lambda_i, \gamma_i$ and $g_i$ are well defined, it will follow from Restriction-Extension Lemma (b) \& (c) that each is \pmf and has topological total variation as indicated in the algorithm, validating steps (3) \& (4) of the  algorithm. 

\begin{lemma}[Restriction-Extension]
\label{restriction-extension}
Choose any closed $D \subset M$ having finitely many components each having nonempty interior, and let the restriction $f' = f|(\mu^{-1}D)$.
\begin{tightList}
\item[(a)]  $f'$ is \pmf and has \mlf $\mu'M'\lambda'$ where  $\mu' = \mu|(\mu^{-1}D)$, $M' = D$, and $\lambda' = \lambda|D$.
\item[(b)] $\ttv(f')$ is equal to the sum of edge lengths in $D$, noting that edge lengths  in $M$ and $M'$ are determined by $\lambda$ and $\lambda'$ respectively.
\item[(c)] Suppose $g$ is a flat extension of $f'$; then $g$ is piecewise monotone and $\ttv(g) = \ttv(f')$.
\end{tightList}
\end{lemma}

\begin{proof}
Because $\mu$ is monotone, $\mu^{-1}D$ has components in 1-1 correspondence with $D$'s components and is therefore a compact subspace,  and thus $f'$ has a monotone-light factorization. 
Every contour of $f'$ is a contour of $f$, so it follows that  $M'$ is homeomorphic to  $D$, and so $\mu' = \mu|(\mu^{-1}D)$ and $\lambda' = \lambda|D$ by uniqueness of monotone-light factorization. $f'$ is \pmf since $M'$ is a finite graph. 
When considered as the middle space of $f'$, $D$ has vertices at what were cut points of edges of $M$. These observations prove (a) \& (b).

We now show (c). Let $\mu_g M_g\lambda_g$ be $g$'s \mlfNoPeriod. For every component $K$ of $\overline{M \smallsetminus D}$, $g$ is constant on $\mu^{-1}K$. Therefore $\mu_g$ maps $\mu^{-1}K$ to a point in $M_g$. Therefore, $M_g$ consists of the components of $D^\circ$,  with some combination of boundary points glued together. This results in a finite graph and therefore $g$ is piecewise monotone. $M_g$ is a quotient of $M$; the quotient map is a homeomorphism on $D^\circ$, from which it follows that $\lambda_g$ measures edge length identically to $\lambda$, and thus $\ttv(g) = \ttv(f')$.
\end{proof}

The following Lemma provides a recipe for computing topological total variation. 

\begin{lemma}[$\ttv$ Decomposition Lemma]
\label{ttv_decomposition}
Let $D_i$, $i = 1 \ldots n$, be any collection of compact subsets covering $M$ that may intersect only at their boundaries.
Then $\ttv(f) = \sum \ttv(\lambda|D_i)$.
\end{lemma}

The $\ttv$ Decomposition Lemma implies  $\ttv(f) = \sum \alpha_i$, where $\alpha_i$ are the amplitudes output by the Varilet Transform Algorithm.
We will use the lemma again in the proof of the Varilet Transform Theorem, in section \ref{proof_part_B}.

\subsubsection{Flat Extension Lemma}
\label{flat_extension}

The following lemma validates step (2) of the Varilet Transform Algorithm, thereby completing the proof that the Varilet Transform Algorithm is well-defined.

\begin{lemma}[Flat Extension Lemma]
\label{flat_extension_lemma}
Choose a varilet lens $\sC=\indexedFull{C}{i}{I}$ of $f$, and let $\sD=\indexed{D_i}$ be the varilet supports.
Choose any $\pi:M\to \Real$ from  $\sC$'s constant-boundary functions.

Then for every varilet support $D_i$, $i \in \sI$:
\begin{tightList}
\item[(a)] There exists a  flat extension $\pi^*:M \to \Real$ of $\pi|D_i$.
\item[(b)] When $M$ is connected then $\pi^*$ is unique.
\item[(c)] $\pi^*$ is a constant-boundary function for $\sC$.
\end{tightList}
\end{lemma}


We will need:

\begin{lemma}
\label{constant_boundary_lemma}
With $\sC$, $\sD$ and $\pi$ as in the previous lemma:
\begin{tightList}
\item[(a)] $\pi$ is constant on the boundary of each component of $\overline{M \smallsetminus D_i}$.
\item[(b)]  Let $z_1 \ldots z_n$ be the unique values among $\pi(\partial C_i)$ and all $\pi(\partial C_j)$, for successors $C_j \in \sS_i$. 
Then the only values that $\pi$ may have on then boundary of any component of  $\overline{M \smallsetminus D_i}$ is one of  $z_1 \ldots z_n$, and for each 
$z_k$ there exists at least one component having boundary upon which $\pi$ has this value.
\end{tightList}
\end{lemma}

\begin{proof}
Let $K = \overline{M \smallsetminus D_i} = \overline{C_i^c} \cup (\underset{\sS_i}{\cup}C_j)$.
Then $\partial K = \partial C_i \cup (\underset{\sS_i}{\cup}\partial C_j)$.

$C_i$'s successors $C_j \in \sS_i$ are disjoint, but  a successor $C_j$ may have boundary intersecting $\partial{C_i}$.
Because each is connected, each successor $C_j$ lies in a single component of $K$.
To prove (a), we must show that when $C_i$'s successors $C_j, C_k$ lie in the same component of $K$, then $ \pi(\partial C_j) = \pi(\partial C_k)$.
But when they lie in the same component, then there must exist a path connecting them in $C_i^c$, from which it follows that $\pi(\partial C_i) = \pi(\partial C_j) = \pi(\partial C_k)$.
Statement (b) follows immediately.
\end{proof}

\begin{proof}[Proof of Flat Extension Lemma (\ref{flat_extension_lemma})]
Using the notation of the preceding proof, 
extend  $\pi|D_i$ to $\pi^*:M \to \Real$ by defining $\pi^*$ on each component of $K' \subset K$ to have constant value $\pi(\partial K')$, thereby proving (a).
Statements (b) \& (c) follow immediately.
\end{proof} 


\newpage

\subsection{Additive Decomposition}
\label{additive_decomposition}

In this section we prove part of the  Varilet Transform Theorem (\ref{transform}).

\begin{lemma}[Additive Decomposition]
\label{additive_decomposition_lemma} 
$f = \sum \alpha_ig_i$.
\end{lemma}

To prove the Additive Decomposition Lemma,  we first \emph{link} constant-boundary function values in a tree structure, and then use an inductive argument exploiting the links. The Link Lemma is proved in the next subsection, with  proof of the Additive Decomposition Lemma in the section following.

\subsubsection{Link Lemma}

The varilet supports $\sD$ are fully determined by choice of varilet lens $\sC$. The inclusions within $\sC$ impart an identical tree structure on $\sD$.
 The next lemma describes how this tree is manifested in the boundaries of the varilet supports in $\sD$, and its implications for constant-boundary functions.

\begin{lemma}[Link Lemma]
\label{link_lemma}
Choose a varilet lens $\sC = \indexedFull{C}{i}{I}$, and let $\sD = \indexed{D_i}$ be the varilet supports.
Choose any index $i \in \sI$.

Then for each successor $C_j \in \sS_i$ there exist points $p_i \in \partial D_i$, $q_j\in  \partial D_j$, such that for any constant-boundary function $\pi$ for $\sC$,
$\pi(p_i) = \pi(q_j)= \pi(\partial C_j)$.
\end{lemma}

We call $(p_i, q_j)$ a \emph{link pair}, because it associates a point in predecessor $D_i$ to a point in successor $D_j$ having equal $\pi$-value, independent of the choice of constant-boundary $\pi$.
There exists a link pair $(p_i, q_j)$ for all indices $j \in \sI$ except the root indices.

\begin{proof}
Recalling definition \ref{supports}, it follows from pairwise disjointness of the $C_i$'s successors that $\partial C_j$ intersects both $\partial D_i$ and $\partial D_j$.
\end{proof}

\subsubsection{Proof of Additive Decomposition}

Using functions $\lambda_i,\gamma_i$  defined in steps (2) \& (3) of the Varilet Transform Algorithm,
we will use the Link Lemma to show $\lambda = \sum \alpha_i\gamma_i$.
Then, using $\mu$ to pull back from $M$ to $f$'s domain $X$,  we get $f = \sum \alpha_ig_i$.

We  use the following:

\begin{lemma}[Zero Varilet Lemma]
\label{zero_varilet_lemma}
Choose any varilet support $D_i$.

Then $\gamma_i|D_j \equiv 0$ for all indices $j$ such that $C_j \not\subset C_i$.
\end{lemma}

\begin{proof}
Step (3) of the Varilet Transform Algorithm states that $\gamma_i(\partial C_i) \equiv 0$.
When $C_j \not\subset C_i$ then $D_j \subset C_i^c$.
Since $\gamma_i$ is constant on each such $D_j$,
the Link Lemma (\ref{link_lemma}) implies $\gamma_i(\partial D_j) = \gamma_i(\partial C_i)  = 0$.
\end{proof}

We now continue with:

\begin{proof}[Proof of Additive Decomposition Lemma (\ref{additive_decomposition_lemma})]


For each index $i \in \sI$ let  
\begin{equation}
\label{pi-def}
\pi_i = \underset{j \le i}{\sum} \alpha_j\gamma_j \text{\ \ \  and\ \ \  }
E_i = \underset{j \le i}{\bigcup}D_j.
\end{equation}

We will show  that $\pi_i | E_i = \lambda|E_i$. Letting $n = \max(\sI)$, since $\pi_n =\lambda$ and $E_n = M$, this will prove $\lambda = \sum \alpha_i\gamma_i$. 

We proceed  by induction on increasing indices $i \in \sI$.

When $i$ is a root index,  steps (2) \& (3) of the Varilet Transform Algorithm state that $\pi_i | D_i = \alpha_i\gamma_i|D_i= \lambda|D_i$.
All indices $j < i$ are also root, so $\pi_i | E_i = \lambda|E_i$.

When $i$ is not a root index, we can assume by induction that $\pi_j | E_j = \lambda|E_j$  for all indices $j<i$.  We must show
\begin{equation*}
\lambda|E_i = \pi_i | E_i.
\end{equation*}
By  the Zero Varilet Lemma (\ref{zero_varilet_lemma}),  $\gamma_i|D_j \equiv 0$ for all $j < i$.
Let $k$ be the index of $C_i$'s predecessor. For index $j$, $k < j < i$, $C_j$ and $C_i$ are disjoint, so $\gamma_j|D_i \equiv 0$, again by  the Zero Varilet Lemma (\ref{zero_varilet_lemma}). 
Therefore, for any $j$, $k \le j < i$,
\begin{align*}
\pi_i | E_j &=  \pi_k | E_j\\
&= \lambda|E_j \text{\ \ \ by induction.}
\end{align*}
Thus, it suffices to show 
\begin{equation*}
\lambda|D_i = \pi_i|D_i.
\end{equation*}
Expanding the definition of $\pi_i$,
\begin{align*}
\pi_i|D_i &= \underset{j \le i}{\sum} \alpha_j(\gamma_j|D_i) \\
&= \alpha_i\gamma_i|D_i + \pi_k|D_i,
\end{align*}
and thus it suffices with
\begin{equation}
\label{expand}
\lambda|D_i =  \alpha_i\gamma_i|D_i + \pi_k|D_i.
\end{equation}
$\pi_k$ is constant on $D_i$; we can identify the constant value: The Link Lemma provides linked pair $(p_k, q_i)$, where $p_k \in \partial D_k$, $q_i \in \partial D_i$, 
and $\pi_k(p_k) = \pi_k(q_i) = \pi_k(\partial C_i)$. 

But by induction, $\pi_k(\partial C_i) = \lambda(\partial C_i)$. Thus equation (\ref{expand}) is equivalent to
\begin{equation*}
\lambda|D_i =  \alpha_i\gamma_i|D_i + \lambda(\partial C_i).
\end{equation*}
The definition of $\gamma_i$ in step (3) of the Varilet Transform Algorithm completes the proof.
 \end{proof}

\subsection{Varilet Basis}
\label{varilet_basis_section}

Proof of the Varilet Transform Theorem (\ref{transform}) is completed in this section.

\begin{lemma}[Varilet Basis]
\label{varilet_basis_lemma}
$\vt{f}{C} = \indexedFull{g}{i}{I}$ is a varilet basis.
\end{lemma}

We start by defining \emph{link recursion}, a definitional schema whereby the Link Lemma (\ref{link_lemma}) enables piecewise definition of constant-boundary functions for $\sC$.
Section \ref{varilet_filter_factor} uses link recursion to define  the \emph{filter factor $\psi$},  proving that it describes varilet filtered versions of $f$.
Section \ref{proof_part_A}  describes how proof of the Valrlet Basis Lemma (\ref{varilet_basis_lemma}) can be split into two cases A \& B, and we prove case A.
The development continues in section \ref{filter_quotient} with definition of the \emph{filter quotient $\phi$}, and identification of the \mlf of varilet filtered functions, thereby enabling proof of  case B, in section \ref{proof_part_B}.

\subsubsection{Link Recursion}

This section describes \emph{link recursion}, a method for piecewise definition of $\pi:M \to \Real$, a new constant-boundary function for $\sC$. Link recursion is based on the Link Lemma (\ref{link_lemma}).

To define a  constant-boundary function 
$\pi:M \to \Real$ using link recursion, we iterate through the indices $i \in \sI$ in increasing order, defining $\pi$ on $D_i$ at each iteration. 

The root indices $i$ come first in the iteration; for each root index $i$, $\pi_i|D_i$ may be independently defined as any \pmf function that is constant on each $\partial C_j$, for $C_i$'s successors $C_j$.

Following the root indices, for each iteration $i$, let $k < i$ be the index of $C_i$'s predecessor, and let $(p_k, q_i)$ be their link pair. $\pi|D_i$ may be independently defined as any \pmf function such that $\pi$ is constant on each $\partial C_j$, for $C_i$'s successors $C_j$, and such that $\pi$ is constant on $\partial C_i$ with value $\pi(\partial C_i) = \pi(p_k)$, noting that the value of $\pi(p_k)$ has previously been defined in iteration $k$. 

Link recursion results in a well-defined, continuous, \pmf function $\pi$, a constant-boundary function for $\sC$.

\subsubsection{Varilet Filter Factor}
\label{varilet_filter_factor}

In this section we discuss varilet filters in relation to the Varilet Transform Algorithm. 

The Varilet Transform Algorithm  results in functions $\vt{f}{C} = \indexed{g_i}$ and amplitudes $\indexed{\alpha_i}$ such that $f = \sum \alpha_ig_i$. 
For filter coefficients $\sA = \indexedFull{a}{i}{I}$, we define notation for the varilet filtered function:
\begin{equation*}
\vf{f}{C}{A} = \sum a_ig_i.
\end{equation*}

To prove that  $\vt{f}{C}$ is a varilet basis, we need to show that $\ttv(\sum a_ig_i) = \sum |a_i|$, for every choice of  $\indexedFull{a}{i}{I}$.
Therefore it should not be surprising that proof of the Varilet Basis Lemma (\ref{varilet_basis_lemma}) uses results that apply to varilet filters generally.

\begin{definition}[Varilet Filter Factor $\psi$]
\label{filter_factor}
Choose a varilet lens $\sC= \indexedFull{C}{i}{I}$, and let $\sD=\indexed{D_i}$ be the varilet supports.  

For any choice of filter coefficients $\sA =\indexed{a_i}$,
the \emph{filter factor} $\psi_\sA:M \to \Real$ is defined by link recursion:
\begin{equation*}
\label{psi-def}
\psi_\sA(p) = 
\begin{cases} 
(a_i/\alpha_i)\lambda(p) &  p \in D_i \text{, when } i  \text{ is a root index; }\\
(a_i/\alpha_i)\big(\lambda(p)-  \lambda(\partial C_{i})\big) + \psi_\sA(\partial C_{i}) &  p \in D_{i} \text{, iterating over increasing } i.
\end{cases}
\end{equation*}
\end{definition}

$\psi_\sA$ is a \pmf  constant-boundary function for $\sC$.

Let $\sF$ be the collection of all indices $j$ such that $a_j = 0$; then $\psi_\sA$ is a flat extension of $D =\psi_\sA|(\underset{j \not\in \sF}{\cup}D_j)$.  

The name ``filter factor'' is motivated by the substitution of $\psi$ for $\lambda$ in $f$'s monotone-light factorization in:

\begin{lemma}[Filter Factor]
\label{filter_factor_lemma} 
Let $\vt{f}{C} = \indexedFull{g}{i}{I}$, let $\sA = \indexed{a_i}$ be any choice of varilet filter coefficients, and let $\psi_\sA$ be the filter factor. 

Then
$\sum a_ig_i = \psi_\sA \circ \mu$.
\end{lemma}

\begin{proof}
Let $\psi = \psi_\sA$, let $f' = \psi \circ \mu$, and let $X_i = \mu^{-1}D_i$ for each $i \in \sI$.
We proceed  by induction on increasing indices $i \in \sI$.

Suppose $i \in \sI$ is a root index; the definition (\ref{filter_factor}) for $\psi$ together with root-index definition of $\gamma_i$ in step (3) of the Varilet Transform Algorithm imply that $f' | X_i = a_ig_i |X_i$. Any $j < i$ are also root indices; therefore  $f' = \underset{j \le i}{\sum}a_jg_j$ on $\underset{j \le i}{\cup}X_j$.

Now suppose $i \in \sI$ is not a root index. We may inductively assume
\begin{equation*}
f' = \underset{j < i}{\sum}a_jg_j \text{\ \ \  on\ \ \ } \underset{j < i}{\bigcup}X_j.
\end{equation*} 
We must show 
\begin{equation*}
f' = \underset{j \le i}{\sum}a_jg_j \text{\ \ \  on\ \ \ } \underset{j \le i}{\bigcup}X_j.
\end{equation*}
By the Zero Varilet Lemma (\ref{zero_varilet_lemma}), $g_i|X_j \equiv 0$ for all $j < i$,  and so by induction it suffices to show
\begin{equation}\label{induction3} 
 f' = \underset{j \le i}{\sum}a_jg_j \text{\ \ \  on\ \ \ }X_i.
 \end{equation}
Let $k < i$ be the index of $C_i$'s predecessor.
$g_j|X_i$ is constant for each $j \le k$, and therefore so is their sum.
Using the Link Lemma (\ref{link_lemma}) we can identify the constant value: 
\begin{align*}
\underset{j \le k}{\sum}a_j(g_j|X_i) &= f'(\partial X_i)\\
&= \psi(\partial X_i)\text{\ \ \ by induction.}
\end{align*}
For  index $j$, $k < j < i$, $C_j$ and $C_i$ are disjoint, and so $g_j|X_i \equiv 0$ by the Zero Varilet Lemma (\ref{zero_varilet_lemma}).
Therefore, expanding the  sum in equation (\ref{induction3}), it is sufficient to show
\begin{equation*}
 f' =  a_ig_i + \psi(\partial X_i)  \text{\ \ \  on\ \ \ }X_i,
\end{equation*}
 which follows  from definition of $\psi$ and the non-root-index definition of $\gamma_i$ in step (3) of the Varilet Transform Algorithm.
\end{proof}

\subsubsection{Proof of Varilet Transform Theorem, Part A}
\label{proof_part_A}

We must show that $\vt{f}{C} = \indexed{g_i}$ is a varilet basis; i.e.\ for every $\indexed{a_i}$
\begin{equation*}
\label{filter_ttv1}
\ttv(\sum a_ig_i) = \sum |a_i|.
\end{equation*}
In light of the Filter Factor Lemma  (\ref{filter_factor_lemma}), this is equivalent to
\begin{equation}
\label{filter_ttv2}
\ttv(\psi \circ \mu) = \sum |a_i|.
\end{equation}

The intuition for the proof is that the filter factor $\psi$ multiplicatively stretches the lengths of the edges in each varilet support $D_i$, using  coefficient $a_i$ as stretch factor.
To measure topological total variation of $f' = \psi \circ \mu = \sum a_ig_i$ we must identify $f'$'s middle space and light factor.
When the coefficients $a_i$ are all nonzero, then the \mlf of $f'$ is $\mu M\psi$, from which we easily show  (\ref{filter_ttv2}).
A small complexity enters the proof when one or more $a_i$ are zero; in this case the middle space of $f'$ is not $M$. We identify the middle space as a certain quotient of $M$; this will allow us to  show (\ref{filter_ttv2}).

We break up the proof of the Varilet Transform Theorem (\ref{transform}) into two parts: Part A covers the case where all filter coefficients $a_i \ne 0$; part B covers the case where one or more $a_i = 0$. 

\begin{proof}[Proof of theorem \ref{transform}, part A ($a_i \ne 0$ for all $i \in \sI$)]
Let $\mu'M'\lambda'$ denote the \mlf of $f' = \sum a_ig_i = \psi \circ \mu$.
Filter factor $\psi$ is light because $\lambda$ is and each $a_i \ne 0$, and therefore by uniqueness of monotone-light factorization  $\mu' = \mu$, $M' = M$, and $\lambda' = \psi$.

Because $M' = M$, $\sC$ is a varilet lens for $f'$. The $\ttv$ Decomposition Lemma (\ref{ttv_decomposition}) states that $\ttv(f') = \sum \ttv(\psi|D_i)$. By definition (\ref{filter_factor}) of $\psi$, 
\begin{align*}
\sum \ttv(\psi|D_i) &= \sum (|a_i|/\alpha_i) \ttv(\lambda|D_i) \\
&= \sum |a_i|,
\end{align*}
completing the proof. 
\end{proof}


\subsubsection{Varilet Filter Quotient}
\label{filter_quotient}

The section identifies the monotone-light factorization of filtered function $f' = \vf{f}{C}{A}$. 

\begin{lemma}[Filter Quotient Lemma]
\label{filter_quotient_lemma}
Suppose \pmf $f$ has \mlf $\mu M\lambda$. 
Choose a lens $\sC = \indexedFull{C}{i}{I}$, and choose  filter coefficients $\sA = \indexed{a_i}$, and suppose varilet filtered $f' = \vf{f}{C}{A}$ has \mlf $\mu'M'\lambda'$.

Then:
\begin{tightList}
\item[(a)] $M'$ is a certain quotient of $M$, described in the proof. 
\item[(b)] Denote the quotient map as $\phi_\sA: M \to M'$. Then $\phi_\sA$ is monotone.
\item[(c)] $\mu' = \phi_\sA \circ \mu$, $M' = \phi_\sA(M)$, and $\lambda' =  \psi_\sA \circ  \phi_\sA^r$, for any right inverse $\phi_\sA^r$.
\end{tightList}
\end{lemma}

\begin{definition}[Filter Quotient $\phi$]
The quotient map $\phi_\sA$ in lemma \ref{filter_quotient_lemma} is called \emph{$\sA$'s filter quotient}.
\end{definition}

\begin{proof}[Proof of lemma \ref{filter_quotient_lemma}]
The statements pertaining to the case where all $a_i \ne 0$, and therefore $M' = M$,  follow from the arguments in the proof of part A of Varilet Transform Theorem in the previous section.
We prove the lemma for the case that one or more $a_i = 0$. Let filter factor $\psi = \psi_\sA$.

Let $\sF$ be the nonempty collection of all indices $j \in \sI$ for which $a_j= 0$, and assume  $\sI \smallsetminus \sF$ is nonempty (since otherwise $\sum a_ig_i \equiv 0$). 

We claim that $M'$ is the quotient of $M$ that identifies all points in each component of $(\underset{j \in \sF}{\cup}D_j)$. Let $\phi = \phi_\sA$ be the quotient map.

We proceed by showing that $\phi$ and any right inverse $\phi^{r}$ make the following diagram commute, recalling that Filter Factor Lemma (\ref{filter_factor_lemma}) states that $f' = \sum a_ig_i = \psi \circ \mu$.
\begin{equation} \label{diag1}
\includegraphics[width=2in]{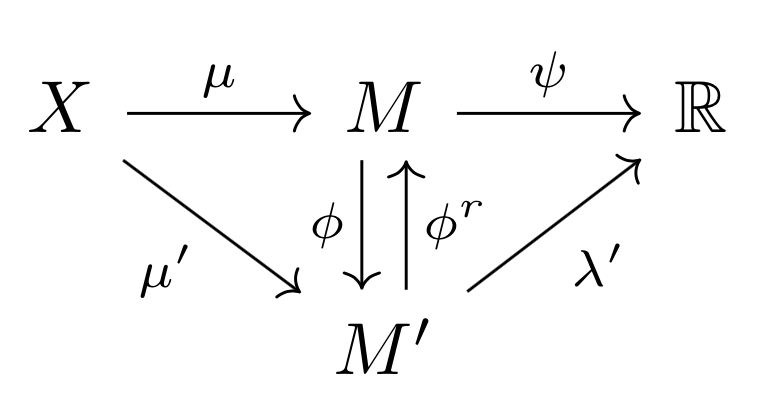}
\end{equation}

It is obvious that $\phi(M)$ is a finite graph. 
We show that $\phi(M)$ is the middle space $M'$ of $f'$ by showing that $\phi \circ \mu $ is monotone and that $\psi  \circ \phi^r$ is light, for any choice of right inverse $\phi^r$.

$ \phi \circ \mu$ is monotone because $ \phi$ is monotone: For any point $q \in \phi(M)$ the set $\phi^{-1}q$ is connected, being either singleton or a connected 
component of $(\underset{j \in \sF}{\cup}D_j)$.

We show that \emph{any} right inverse  $\phi^{r}$ can be chosen.
Let $p$ be any point chosen from $\phi^{-1}q$; we show that the value of $\psi(p)$ does not depend on the choice.
When $\phi^{-1}q$ is singleton then the choice is irrelevant. 
Otherwise, $\phi^{-1}q$ is a connected  component of $(\underset{j \in \sF}{\cup}D_j)$, and $\psi$ is constant on each such component.

To see that $ \psi \circ \phi^r $ is light, choose any $q \in \phi(M)$, and let $z =  (\psi \circ \phi^r)(q)$. We show that $(\psi \circ \phi^r)^{-1}z$ is a finite collection of points. Because $\phi(M)$ is a finite graph, if $(\psi \circ \phi^r)^{-1}z$ were infinite then it would contain an open set of $\phi(M)$, in which case $\psi^{-1}z$ must contain an open set of $M$. But any open set of  $M$ upon which $\psi$ is constant must lie in 
$(\underset{j \in \sF}{\cup}D_j)$, which is mapped by $\phi$ to a finite set of points in $\phi(M)$, a contradiction.

We have now confirmed that diagram (\ref{diag1})  commutes, i.e.\ that the \mlf of $f'$ is $\mu' = \phi \circ \mu$, $M' = \phi(M)$, and $\lambda' = \psi \circ \phi^r$.
\end{proof}

We state some immediate consequences regarding the filter quotient $\phi$:

\begin{Lemma}[Filter Quotient Lemma]
\label{filter_quotient_Lemma}
\begin{tightList}
\item
\item[(a)] When $a_i \ne 0$ then:
\begin{tightList}
\item[-] $\phi$ is a homeomorphism on $D_i^\circ$. 
\item[-] $\phi(D_i)$ may have fewer components than did $D_i$, but may not have more. 
\item[-] $\phi^{-1}(\partial \phi(D_i)) \subset \partial D_i$, but $\phi$ may also map points in $\partial D_i$  to $\phi(D_i)^\circ$.
\end{tightList}
\item[(b)] The collection $\{\phi(D_i)\ |\  a_i \ne 0\}$ covers $M'$ and the sets $\phi(D_i)$ intersect only at their boundaries. 
\end{tightList}
\end{Lemma}

\subsubsection{Proof of Varilet Transform Theorem, Part B}
\label{proof_part_B}

It is now easy to complete the proof of the Varilet Transform Theorem (\ref{transform}).

\begin{proof}[Proof of theorem \ref{transform}, part B]
Let $f' = \vf{f}{C}{A}$; we must show that $\ttv(f') = \sum |a_i|$.

The Filter Quotient Lemma (\ref{filter_quotient_Lemma}) and the $\ttv$ Decomposition Lemma (\ref{ttv_decomposition}) allow us to compute $\ttv(f')$,
\begin{equation}
\label{sum_for_theorem}
\ttv(f') = \underset{a_i \ne 0}{\sum} \ttv(\lambda'|\phi(D_i))
\end{equation}

When $a_i \ne 0$, since $\phi$ is a homeomorphism on $D_i^\circ$, it follows from the Filter Quotient Lemma (\ref{filter_quotient_lemma}) and definition (\ref{filter_factor}) of filter factor $\psi$ that
\begin{align*}
\ttv(\lambda'|\phi(D_i)) &= \ttv(\psi|D_i)\\
&=(|a_i|/\alpha_i)\ttv(\lambda|D_i)\\
&=|a_i|.
\end{align*}

Thus, equation (\ref{sum_for_theorem}) becomes
\begin{align*}
\ttv(f') &= \underset{a_i \ne 0}{\sum} |a_i|\\
&= \underset{\sI}{\sum} |a_i|,
\end{align*}
completing the proof.
\end{proof}

\bibliography{varilets}

\begin{thebibliography}{10}

\bibitem{Bauer2014}
Ulrich Bauer, Axel Muk, Hannes Sieling, and Max Wardetzky.
\newblock Persistent homology meets statistical inference -- a case study:
  Detecting modes of one-dimensional signals, 2014.

\bibitem{Bauer2010}
Ulrich Bauer, Carola-Bibiane Sch{\"o}nlieb, and Max Wardetzky.
\newblock Total variation meets topological persistence: A first encounter.
\newblock In {\em Proceedings of ICNAAM 2010}, pages 1022--1026, 2010.

\bibitem{Bremer2004}
P.-T. Bremer, B.~Hamann, H.~Edelsbrunner, and V.~Pascucci.
\newblock A topological hierarchy for functions on triangulated surfaces.
\newblock {\em IEEE Transactions on Visualization and Computer Graphics},
  10(4):385 -- 396, 2004.

\bibitem{varilet_image}
Martin Brooks.
\newblock Persistence lenses: Segmentation, simplification, vectorization,
  scale space and fractal analysis of images, 2016.
\newblock arXiv:1604.07361.

\bibitem{Carlsson2009}
Gunnar Carlsson.
\newblock Topology and data.
\newblock {\em Bull. Amer. Math. Soc.}, 46:255--308, 2009.

\bibitem{Carlsson2004}
Gunnar Carlsson, Afra Zomorodian, Anne Collins, and Leonidas Guibas.
\newblock Persistence barcodes for shapes.
\newblock {\em Eurographics Symposium on Geometry Processing}, 2004.

\bibitem{Carr2004}
Hamish Carr.
\newblock Topological manipulation of isosurfaces.
\newblock PhD Thesis, University of British Columbia, 2004.

\bibitem{Carr2000}
Hamish Carr, Jack Snoeyink, and Ulrike Axen.
\newblock Computing contour trees in all dimensions.
\newblock In {\em Proceedings of the Eleventh Annual ACM-SIAM Symposium on
  Discrete Algorithms}, SODA '00, pages 918--926, Philadelphia, PA, USA, 2000.
  Society for Industrial and Applied Mathematics.

\bibitem{Charatonik1998}
J.~J. Charatonik.
\newblock History of continuum theory.
\newblock In C.~E. Aull and R.~Lowen, editors, {\em Handbook of the History of
  General Topology, vol. 2}, pages 703--786. Kluwer Academic Publishers, 1998.

\bibitem{Edelsbrunner2007}
D.~Cohen-Steiner, H.~Edelsbrunner, and J.~Harer.
\newblock Stability of persistence diagrams.
\newblock {\em Discrete Comp. Geo.}, 37:103 -- 120, 2007.

\bibitem{EdelsbrunnerBook}
Herbert Edelsbrunner and John Harer.
\newblock {\em Computational Topology - an Introduction.}
\newblock American Mathematical Society, 2010.

\bibitem{Edelsbrunner2002}
Herbert Edelsbrunner, David Letscher, and Afra Zomorodian.
\newblock Topological persistence and simplification.
\newblock {\em Discrete Comp. Geo.}, 28:511 -- 533, 2002.

\bibitem{Forman01}
Robin Forman.
\newblock A user's guide to discrete morse theory.
\newblock In {\em Proc. of the 2001 Internat. Conf. on Formal Power Series and
  Algebraic Combinatorics, A special volume of Advances in Applied
  Mathematics}, page~48, 2001.

\bibitem{Georgakopoulos}
Agelos Georgakopoulos.
\newblock On graph-like continua of finite length.
\newblock arXiv:1401.5946, 2014.

\bibitem{TVdefs}
Boris~I. Golubov and Anatolii~G. Vitushkin.
\newblock Variation of a function.
\newblock In Michiel Hazewinkel, editor, {\em Encyclopedia of Mathematics}.
  Springer, 2001.

\bibitem{Guillemard}
Mijail Guillemard and Armin Iske.
\newblock Signal filtering and persistent homology: An illustrative example.
\newblock In {\em Proc. SampTA}, 2011.

\bibitem{Guillemard2013}
Mijail Guillemard, Gitta Kutyniok, Holger Boche, and Friederich Philipp.
\newblock Signal analysis with frame theory and persistent homology, 2013.

\bibitem{HY1961}
John~G. Hocking and Gail~S. Young.
\newblock Topology.
\newblock {\em Addison-Wesley Publishing Co., Inc., Reading, Mass.-London},
  1961:ix+374, 1961.

\bibitem{Nadler1992}
Sam B.~Nadler Jr.
\newblock {\em Continuum Theory, An Introduction}.
\newblock Marcel Dekker, Inc., New York, 1992.

\bibitem{Lord1997}
Harriet Lord.
\newblock Monotone-light factorizations: a brief history.
\newblock Preprint available at \url{/www.csupomona.edu/~hlord/hh60.ps}.

\bibitem{Reeb1946}
Georges Reeb.
\newblock Sur les points singuliers d'une forme de {P}faff compl\`etement
  int\'egrable ou d'une fonction num\'erique.
\newblock {\em C. R. Acad. Sci. Paris}, 222:847--849, 1946.

\bibitem{Robinson}
Michael Robinson.
\newblock {\em Topological Signal Processing}.
\newblock Springer, 2014.

\bibitem{Rudin}
Leonid~I. Rudin, Stanley Osher, and Emad Fatemi.
\newblock Nonlinear total variation based noise removal algorithms.
\newblock {\em Physica D: Nonlinear Phenomena}, 60:259 -- 268, 1992.

\bibitem{Tierny2012}
J.~Tierny and V.~Pascucci.
\newblock Generalized topological simplification of scalar fields on surfaces.
\newblock {\em Visualization and Computer Graphics, IEEE Transactions on},
  18(12):2005--2013, Dec 2012.

\bibitem{walker1974}
JR~Walker.
\newblock A simplicial monotone-light factorization theorem.
\newblock {\em Fundamenta Mathematicae}, 85(3):229--233, 1974.

\bibitem{DW1979}
Gordon Whyburn and Edwin Duda.
\newblock {\em Dynamic topology}.
\newblock Springer-Verlag, New York, 1979.
\newblock Undergraduate Texts in Mathematics, 
  Kelley.

\bibitem{Whyburn1942}
Gordon~Thomas Whyburn.
\newblock {\em Analytic {T}opology}.
\newblock American Mathematical Society Colloquium Publications, v. 28.
  American Mathematical Society, New York, 1942.

\bibitem{Whyburn1970}
G.T. Whyburn.
\newblock Dynamic topology.
\newblock {\em American Mathematical Monthly}, 77:556--570, 1970.

\bibitem{graphTV}
D.~Uminsky X.~Bresson, T.~Laurent and J.~H. von Brecht.
\newblock An adaptive total variation algorithm for computing the balanced cut
  of a graph, 2013.
\newblock arXiv:1302.2717.

\end{thebibliography}
\bibliographystyle{plain}

\end{document}